\documentclass[11pt]{article}
\usepackage{amsmath,amscd} 
\usepackage{amssymb}
\usepackage{amsthm}
\usepackage{fullpage}

\usepackage{algorithmicx}

\usepackage{algpseudocode}

\newtheorem{theorem}{Theorem}[section]
\newtheorem{proposition}[theorem]{Proposition}

\newtheorem{lemma}[theorem]{Lemma}

\newtheorem{corollary}[theorem]{Corollary}
\newtheorem{myconj}[theorem]{Conjecture}

\newcommand{\leftr}{\mathtt{left}}
\newcommand{\rightr}{\mathtt{right}}

\def\nat{{\mathbb N}}
 \def\real{{\mathbb R}}
 \def\rat{{\mathbb Q}}

\begin{document}

\title{A note on the complexity of comparing succinctly represented integers,
with an application to maximum probability parsing}
\author{Kousha Etessami\\U. of Edinburgh\\{\tt kousha@inf.ed.ac.uk}
\and 
Alistair Stewart\\U. of Edinburgh\\{\tt stewart.al@gmail.com} 
\and
Mihalis Yannakakis\\Columbia U.\\{\tt mihalis@cs.columbia.edu}}

\date{}

\maketitle

\begin{abstract}
\noindent The following two decision problems capture
the complexity of comparing integers or rationals that are 
succinctly represented
in product-of-exponentials notation, or equivalently, 
via arithmetic circuits using
only multiplication {\small \&} division gates, and integer inputs:

\vspace*{0.05in}

\noindent {\bf Input instance:} four lists of positive integers:

\vspace*{0.05in}

$a_1,\ldots, a_n \in \nat_+^n$; \ $b_1,\ldots,b_n \in \nat_+^n$;   
\ $c_1,\ldots,c_m \in \nat_+^m$;
\ $d_1,\ldots,d_m \in \nat_+^m$;

\vspace*{0.04in}

\noindent where each of the integers is represented in binary.\\
  
\noindent {\bf  Problem 1 (equality testing):}  Decide whether  
$a_1^{b_1} a_2^{b_2} \ldots a_n^{b_n} = c_1^{d_1} c_2^{d_2}\ldots c_m^{d_m}$.\\

\noindent {\bf Problem 2  (inequality testing):} Decide whether  
$a_1^{b_1} a_2^{b_2} \ldots a_n^{b_n}  \geq c_1^{d_1} c_2^{d_2}\ldots c_m^{d_m}$.\\

Problem 1 
is easily decidable in polynomial time using a simple iterative
algorithm.
Problem 2 is much harder. 
We observe that the complexity of Problem 2 is intimately 
connected
to deep conjectures and results in number theory.
In particular,
if a refined form of the {\em ABC conjecture}  
formulated by Baker in 1998
holds, or if the older {\em Lang-Waldschmidt
conjecture} (formulated in 1978) on linear forms in logarithms 
holds, then Problem 2 is decidable in P-time 
(in the standard Turing model of computation).
Moreover,
it follows from
the best available quantitative bounds on
linear forms in logarithms, e.g., by Baker and W\"{u}stholz
\cite{BakerWust93} or Matveev \cite{Matveev-2000II},
that
if $m$ and $n$ are fixed universal constants then Problem 2
is decidable in P-time (without relying on any conjectures).
This latter fact was observed earlier by Shub (\cite{Shub-smalefest-93}).

We describe one application: P-time  
maximum probability parsing for {\em arbitrary} stochastic
context-free grammars (where $\epsilon$-rules are allowed).

\end{abstract}

\section{Introduction}
For many computations involving 
large integers, or large/small non-zero rationals,  
it is convenient to be able 
to manipulate and compare the numbers 
without having to 
compute a standard binary representation of them.
Indeed, in many settings it is intractable
to compute such a binary representation.
This has motivated compact representations such as classic floating point, but floating point numbers also suffer a loss
of information (precision), which one would like to avoid if possible.

There are a number of succinct representations 
one could consider for such a purpose.
One approach is to represent integers via
arithmetic circuits (straight-line programs), 
with gates $\{+,-,* \}$, and with integer inputs (represented in binary). 
However, the problem of deciding
whether an integer represented via an arithmetic circuit
 is $\geq$ another such integer, referred to as 
{\bf PosSLP}  \cite{ABKM06}, captures
{\em all} of polynomial time in the unit-cost arithmetic RAM model
of computation.  The best complexity upper bounds we know for PosSLP
in the standard Turing model of computation is
the 4th level of the {\em counting hierarchy}, $_{{\Huge{P}}}{PP^{PP^{PP}}}$, as established by Allender, B\"{u}rgisser, Kjeldgaard-Pedersen,
and  Miltersen in \cite{ABKM06}.
PosSLP subsumes other hard problems whose complexity
is open, like the well-known {\em square-root-sum} problem (\cite{GGJ76}),
and it appears highly unlikely that one could show that PosSLP 
is even in NP.
On the other hand, as noted in \cite{ABKM06}, the problem of testing {\em equality} of 
integers represented by two such arithmetic circuits, {\bf EquSLP},
is P-time equivalent to {\em polynomial identity testing}, and can 
be decided in coRP.  However, it remains open whether EquSLP is in NP 
and showing this would already imply hard circuit lower bounds
(\cite{KI03}), so it is likely to be difficult.
On the other hand, there are no hardness results
known for PosSLP with respect to standard 
complexity classes, beyond
$P$-hardness.  In particular, PosSLP is not known to be NP-hard.

Note that if
 the arithmetic in the computation
is confined to only {\em linear} operations $\{+,-\}$
on registers, and
multiplication by constants in the input, then
the encoding length of the numbers is 
linear in the number of arithmetic operations, so
we can represent all the numbers exactly, 
and P-time in the Turing model can simulate
polynomial time unit-cost {\em linear} arithmetic computation.

Now consider another natural restricted class of
arithmetic circuits, which turn out to 
be useful in a number of settings:
arithmetic circuits containing
only gates $\{ * , / \}$. 
An essentially equivalent representation is the following:
For a list of rational numbers $a = \langle a_1, \ldots, a_n \rangle$, and
a list of integers $b = \langle b_1, \ldots, b_n \rangle$,
both of dimension $n$, 
we use $a^b$ as a shorthand notation for:
$a_1^{b_1} a_2^{b_2} \ldots a_n^{b_n}$.
We shall refer to this succinct representation of integers and rationals
as {\bf\em product of exponentials} ({\bf PoE})  representation.
PoE representation is easily
seen to be equivalent to representation via 
arithmetic circuits with integer
inputs given in binary and with multiplication and division gates $\{ * , / \}$.
The following is shown in the appendix.

\begin{proposition}  
There is a simple P-time
translation from a given number represented in PoE to
the same number represented as an arithmetic 
circuit over $\{*,/\}$ with integer inputs (represented in binary).
Likewise, there is a simple P-time translation in the other direction.
\label{prop:spe-is-same-as-arith-circ}
\end{proposition}

Consider the problem of deciding whether
one rational number, $a^b$, 
given in PoE representation, is greater than (or, respectively, equal to) 
another rational number $c^d$, given
in PoE.  We remark that, again, the inequality testing problem 
basically captures the power of polynomial time in the unit-cost
arithmetic RAM model of computation, where the only arithmetic
operations permitted are $\{*,/\}$.

\vspace*{0.03in}

\noindent {\bf Input instance:} four lists of positive integers:\\  
\mbox{} $\quad \quad \quad$ $a_1,\ldots, a_n \in \nat_+^n$; \ $b_1,\ldots,b_n \in \nat_+^n$;   
\ $c_1,\ldots,c_m \in \nat_+^m$;
\ $d_1,\ldots,d_m \in \nat_+^m$;\\ 
where each of the integers is represented, as usual, in binary.

\vspace*{0.03in}  

\noindent {\bf  Problem 1 (equality testing):}  Decide whether or not $a^b = c^d$.

\vspace*{0.02in}

\noindent {\bf Problem 2 (inequality testing):} Decide whether or not $a^b \geq c^d$.

\vspace*{0.03in}

\noindent Note that, by rearrangement, these problems are equivalent to 
the versions allowing bases $a_i$ and $c_j$  to be
rationals (encoded in binary), and allowing $b_i$ and $d_j$
to be any integers (in binary).

Let us note that PoE representation is in fact already widely used in practice.
Specifically, 
because iterated multiplication may make
numbers very small or very large, practitioners often explicitly recommend using a 
{\em logarithmic transformation} to represent numbers such as 
$a_1^{b_1} \ldots a_n^{b_n}$  by:
$$ b_1 \log(a_1) + b_2 \log(a_2) + \ldots + b_n \log(a_n)$$

This allows multiplications and divisions to be carried
out by using only addition on the coefficients of the log representations.
Note that such ``log representations''
are basically equivalent to PoE, as long as the
logarithms are only interpreted {\em symbolically}.
(Of course, one still needs to be able to {\em compare}
 such sums of logs, and we will return to this shortly.)
One setting where log transformation 
is recommended in practice
is for the analysis of Hidden Markov Models (HMMs) using the Viterbi algorithm, and 
for probabilistic parsing.
For example, the book Durbin et. al. \cite{DEKM99} (section 3.6) says:
\begin{quote}
{\em On modern floating point processors we will run into numerical
problems when multiplying many probabilities in the Viterbi, forward,
or backward algorithms.  For DNA for instance, we might want to model genomic 
sequences of 100 000 bases or more.  Assuming that the product
of one emission and one transition probability is $0.1$, the
probability of the Viterbi path would then be on the order of 
$10^{-100000}$.  Most computers would behave badly with such numbers........
For the Viterbi algorithm we should always use the logarithm 
of all probabilities.  Since the log of a product is the sum of the logs, all the
products are turned into sums.........
It is more efficient to take the log of all of the model parameters
before running the Viterbi algorithm, to avoid calling the
logarithm function repeatedly during the dynamic programming
iterations. } (\cite{DEKM99}, pages 78-79.)
\end{quote}

We justify these comments from a complexity-theoretic viewpoint.
In fact, we do so
in the more general context of computing a maximum probability
parse tree for a given string and given stochastic context-free grammar (SCFG), which generalizes the Viterbi
algorithm for finite-state HMMs.
We will observe that if deep conjectures in number theory
hold then Problem 2 can be solved in polynomial time
by employing the PoE (or log) representation,
and also
that the PoE
representation can be used to  obtain P-time algorithms
for computing a maximum probability parse tree for 
a given string with a given SCFG, and for solving related problems.

We first show Problem 1  is decidable in P-time using an
easy iterative algorithm.
Problem 2 is much harder.
We observe that if the {\em Lang-Waldschmidt
conjecture} \cite{LangDio78} holds, then Problem 2 is decidable in P-time.
Likewise, if Baker's refinement   \cite{Baker-abc-1998} of the {\em ABC conjecture}  of  Masser and Oesterl\`{e}
holds, then again it implies Problem 2 is decidable in P-time. 
The ABC conjecture is considered
one of the central conjectures of modern number 
theory (see, e.g.,  \cite{granville-abc-2002,baker-wustholz-survey,Waldschmidt-survey04,baker-wustholz-book07}). 

Furthermore, we observe that
the best currently known  quantitative bound in Baker's theory of
linear forms in logarithms, e.g., those 
due to Baker and W\"{u}stholz \cite{BakerWust93} or
Matveev \cite{Matveev-2000I,Matveev-2000II}, 
yield that when $m$ and $n$ are fixed universal constants, 
Problem 2 is decidable in polynomial time.     
We note that Shub 
has already observed this fact 
in \cite{Shub-smalefest-93} (Theorem 6, page 454), 
namely that for fixed constants $m$ and $n$,
Problem 2 is decidable in P-time.\footnote{We thank one of the
anonymous referees for bringing 
Shub's paper \cite{Shub-smalefest-93} to our attention.}
Although Shub stated a correct theorem, and sketched a proof based
on the same ideas, the proof in \cite{Shub-smalefest-93}
contains some inaccuracies.   In particular, it mis-states
the lower bounds for linear forms in logarithms.  
In fact, the lower bound 
quoted in \cite{Shub-smalefest-93} is false, 
as we shall explain in a footnote to Proposition \ref{prop:main-comp-prop}.
For this reason, in Proposition \ref{prop:main-comp-prop} we provide a proof
of this result, first observed by Shub \cite{Shub-smalefest-93},
using the best currently available
bounds on linear forms in logarithms (\cite{BakerWust93,Matveev-2000I,Matveev-2000II}).

It is well known that the ABC conjecture, and
related conjectures involving
explicit bounds for linear forms in logarithms,
have important applications for {\em effective solvability} 
of various diophantine equations.
However, to our knowledge, 
it has not been observed previously that
these number-theoretic conjectures 
are also connected to  
the {\em polynomial time solvability}
of basic problems such as the comparison of 
succinctly represented numbers.

We give one application to maximum probability parsing
for {\em stochastic context-free grammars} (SCFG).
Computing the maximum probability (most likely) parse of a given string 
for a SCFG is a
basic task in statistical natural language processes (NLP) \cite{ManSch99}.
Until now, it was only known to be 
computable in P-time for particular classes of SCFGs,
in particular  SCFGs in Chomsky Normal Form, and 
SCFGs not containing arbitrary $\epsilon$-rules.
For general SCFGs, $G$, 
the maximum probability of a parse tree for even a fixed string, $w$, may
be as small as $1/2^{2^{|G|}}$, where $|G|$ is the encoding size
of the SCFG.  Thus one cannot  express such probabilities
in P-time in  binary representation.
However, the maximum parse tree probabilities 
can be expressed concisely in PoE, and if we can check 
inequalities on such encodings of rational 
numbers efficiently, then we can compute the maximum 
parse tree probability in P-time.
 
Specifically, we show that if Baker's refined ABC conjecture holds, 
or if the Lang-Waldschmidt
conjecture holds, then
given an arbitrary SCFG $G$,
and given an arbitrary string $w \in \Sigma^*$,
there is a polynomial-time algorithm that:
(1) computes a (succinct
representation of) a {\em maximum probability
parse tree} for the string $w$ and
also computes (a succinct representation of) the exact 
maximum parse probability 
$p^{\max}_{G,w}$; and
(2) given also a rational probability $q$ given in binary (or in PoE),
decides whether the maximum parse probability of $w$,
$p^{\max}_{G,w}$, satisfies $p^{\max}_{G,w} \geq q$;
(3) given also another string $w' \in \Sigma^*$, decides
whether $p^{\max}_{G,w} \geq p^{\max}_{G,w'}$.  
Furthermore, when the SCFG has
a fixed constant number, $m$, of distinct probabilities 
labeling its rules, all of the above problems (1) -- (3) are in P-time
(in the Turing model), 
without assuming any number theoretic conjectures.
Finally, we show that essentially the same algorithm can
be used to approximate the maximum
parsing probability and compute (a succinct representation of)
an $\epsilon$-optimal parse tree for a string $w$ and a SCFG $G$ 
in time polynomial in the size of 
$G, w$ and $\log(1/\epsilon)$, without assuming any conjectures.

\section{Deciding equality of succinct integers in P-time}
\label{sec:eq}

We now give a simple iterative P-time algorithm for Problem 1 (equality testing).
The algorithm is in Figure \ref{fig:alg1}.  It simply repeatedly computes $g_{i,j} = GCD(a_i,c_j)$ for 
pairs  $a_i$ and $c_j$, and if $g_{i,j} > 1$, then it does the appropriate
adjustments on the succinct representations.   It terminates
when $g_{i,j} = 1$ for all $i,j$.   
The two remaining numbers are $1$ if and only if the original
numbers were equal.

\begin{figure}
\noindent {\bf Input:} 4 lists $a$, $b$, $c$, $d$, of positive integers,
$a$ \& $c$
not containing $1$,  $|a| = |b|$, \& $|c| = |d|$.

\vspace*{0.05in}

\noindent {\bf Task:} Decide whether or not $a^b = c^d$.

\vspace*{0.05in}
\noindent \begin{algorithmic}

\While{there is some $i \in [a]$ and $j \in [c]$ such that 
$GCD(a_i,c_j) > 1$}
\State  Let $i \in [a]$ and $j \in [c]$ be such that
 $g_{i,j} = GCD(a_i,c_j)  > 1$;

\State Let $a_i  \gets a_i/g_{i,j}$; \  Let $c_j \gets c_j/g_{i,j}$;

\If{$a_i = 1$}
\State Remove $a_i$ from list $a$, and remove $b_i$  from list $b$
\EndIf

\If{$c_j =1$}
\State  Remove $c_j$ from list $c$, and remove $d_j$ from list $d$
\EndIf

\If{$b_i > d_j$}

\State  Append integer $g_{i,j}$ to the end of list $a$;
\State  Append integer $b_i - d_j$ to the end of list $b$;

\ElsIf{$b_i <  d_j$}

\State  Append integer $g_{i,j}$ to the end of list $c$;
\State  Append integer $d_j - b_i$ to the end of list $d$;

\EndIf

\EndWhile

\If{$a$ and $c$ are both the empty list}

\Return   EQUAL

\Else

\Return   NOT-EQUAL
\EndIf

\end{algorithmic}

\caption{Simple P-time algorithm for Problem 1 \label{fig:alg1}} 
\end{figure}

In more detail,
at the start of the iterative algorithm, we initialize  four lists
of positive integers:
$a := \langle a_1,\ldots,a_n \rangle$, $b := \langle b_1, \ldots, b_n \rangle$
, $c := \langle c_1, \ldots, c_m \rangle$, and $d :=  \langle d_1, \ldots, d_m \rangle$.
We can assume wlog that lists
$a$ and $c$ do no contain the number $1$.
For a list $g$, we use $|g|$ to denote its length.
So, initially, $|a| = |b| = n$ and $|c| = |d| = m$.
Every iteration of the algorithm will maintain the
following invariant: $|a| = |b|$ and $|c| = |d|$. 
For a list $g$, we let $[g] = \{1, \ldots, |g|\}$.
For a list $g$, and for $i \in [g]$ we use $g_i$ to
denote the $i$'th element of the list $g$. 
Given two lists of integers, $g$ and $f$, where $|g| = |f| = r$,
{we use $g^f$ to denote $g_1^{f_1} \ldots g_r^{f_r}$}.

\begin{proposition}
\label{prop:alg1}
The algorithm in Figure 1 decides Problem 1, and runs in
polynomial time.\footnote{It is worth noting that a more careful
implementation of this simple algorithm, based on existing 
``factor refinement'' procedures (\cite{BachDS93,bernstein05}),
can be used to solve Problem 1 very efficiently.} 
\end{proposition}

\begin{proof}
For the correctness of the algorithm, 
we first observe that the following invariant is maintained: the equality
\begin{equation}\label{eq:basic}
a^b = c^d
\end{equation} 
holds before an iteration
of the while loop {\em if and only if} it holds after an iteration
of the while loop.

To see why this is the case, suppose
that $GCD(a_i,c_j) = g_{i,j} > 1$, and suppose
that an iteration of the while loop is conducted using
this $i$ and $j$.   Then it is clear that the new lists
$a$, $b$, $c$, and $d$ are obtained by 
simply factoring $g_{i,j}$ out of $a_i$ and $c_j$ on the
left and right side of equation (\ref{eq:basic}), and then
dividing both sides of the equation (\ref{eq:basic}) by 
$g_{i,j}^{d_j}$  or by $g_{i,j}^{b_i}$, depending,
on whether $b_i \geq d_j$, or $d_j \geq b_i$, respectively.
Note that when
$b_i = d_j$, dividing both sides by $g_{i,j}^{b_i}$ eliminates
this power of $g_{i,j}$ from both sides, so we do not need
to include a power of $g_{i,j}$ on either side.
Since factoring out $g_{i,j}$ and dividing both sides of the
equation by a positive value are reversible, we conclude
that the invariant is maintained.

Let us now argue that if the algorithm halts,
i.e.,
if the while loop terminates, then the output is correct.
Indeed, if the while loop terminates this
 means for every
pair of numbers $a_i$ and $c_j$ in the lists $a$ and $c$ 
we have $GCD(a_i,c_j) = 1$.   Thus we also have  
$GCD(a^b , c^d) = 1$.   But in that case if either $a^b \neq 1$ or $c^d \neq 1$,
then $a^b \neq c^d$.  Thus if the while loop halts the algorithm correctly
returns ``EQUAL'' or ``NOT-EQUAL'' depending on whether or not $a^b = c^d$ for
the original input lists.

The only thing left to establish is that the algorithm always halts
and runs in polynomial time.

For a list of positive integers $a$,  not containing the number $1$,
let us call another list $a'$ of positive integers a {\em factored subform}
of $a$, if there is a mapping $\phi: [a'] \rightarrow [a]$
that maps indices $r \in [a']$ to indices $\phi(r) \in [a]$,
such that for all $i \in [a]$

$$ (\prod_{r \in \phi^{-1}(i)} a'_r)  \mid  a_i $$

In other words, the product of all entries of $a'$ that map to 
the entry $a_i$ of $a$ should divide $a_i$. 

We shall call $a'$ a {\em non-trivial} factored subform of $a$
if the list $a'$ does not contain any entries equal to $1$ either, 
and furthermore no permutation of the list $a'$ is identical to the list $a$.

Next, let us observe that after each iteration of the While
loop, the positive integer lists $a$ and $c$ 
must each be non-trivial factored subforms of the 
respective positive integer lists 
$a$ and $c$ prior to that iteration of the while loop.
Thus, by induction on the number of iterations, after any number of 
iterations of the while loop we must have  $a$ 
and $c$ consisting of non-trivial factored
subforms of the original input lists $a$ and $c$ of positive integers.

But the while loop must therefore terminate, because by the fundamental
theorem of arithmetic (unique prime factorization) 
there can not exist an unbounded sequence
of lists of positive integers 
$$a^0, a^1, a^2 , a^3, \ldots$$ 
such that each one does not contain the number $1$, 
and such that for all $i \in \nat$
$a^{i+1}$ is a non-trivial factored
subform of $a^{i}$.

Furthermore, for the same reason, the while loop must terminate
after a number of iterations that is polynomial in the encoding
size of the original lists $a$ and $c$.  Namely, the number
of iterations can not be greater than the number of prime
factors of the integers in the lists $a$ and $c$.

Furthermore, the encoding size of the lists $a$ $b$, $c$, and  $d$,
always remains polynomial in the encoding size of the original
input lists.  This is so, firstly, because $a$ and $c$ 
always remain factored subforms of the original lists, respectively,
and furthermore
because the maximum value of the positive integers in lists $b$ and $d$
(which are always the same size as their corresponding lists $a$ and $c$),
is never more than their maximum value in the original lists $b$ and $d$,
respectively.

It is then clear that each iteration can be carried
out in time polynomial in the encoding size of the original lists,
because each iteration of the while loop, when the
current lists given by $a$, $b$, $c$, and $d$, requires at most $|a|*|c|$
computations of GCDs on pairs of integers that are no bigger
than the maximum integer in the original lists, and 
as already established at any iteration the lists themselves
are only polynomially bigger than the original lists.
 \end{proof}

\section{Deciding inequalities between succinct integers}

\label{sec:ineq}

We now consider the much harder Problem 2 (inequality testing).
We first recall a deep theorem due to Baker and W\"{u}stholz
on explicit quantitative bounds on linear forms in logarithms.\footnote{The general theorem regards logarithms
of algebraic numbers.  We will state it in the special
case of logarithms of standard integers, which suffices for
our purposes.}
Let $a_1,\ldots, a_n$ be positive integers, none
of which are equal to $0$ or to $1$,  let $b_1, \ldots, b_n$
be arbitrary integers not all equal to $0$.  Let  $e$ be the base of the natural
logarithm, and define $B := \max \{ |b_1| , |b_2|, \ldots, |b_n|, e\}$. 
Let  
$h'(a_i) = \max \{  \log a_i , 1 \}$,   
where $\log$ denotes
the natural logarithm.  Let
$$ \Lambda(a,b)  :=  \log(a^b) = b_1 \log a_1 + b_2 \log a_2 + \ldots + b_n \log a_n $$
For any list $a$ of positive integers, and list $b$ of integers, both of length $n$, let
$$G(a,b) := e^{-C(n) h'(a_1) h'(a_2) \ldots h'(a_n) \log B}$$
where $C(n) := 18 (n+1)! n^{n+1} 32^{n+2} \log(2n)$.

\begin{theorem}[(Baker-W\"{u}stholz, 1993 \cite{BakerWust93})]
\label{thm:bak-wust}
For any list $a$ of positive integers 
and any list $b$ of non-zero integers\footnote{
Although this will be obvious to experts, let us
explain why this theorem is indeed a specialization 
(to positive integer $a_i$'s) of Baker-W\"{u}stholz's.  
 The Theorem
in \cite{BakerWust93} allows $a_i$'s 
to be algebraic numbers, and $\log a_i$ is defined
to be {\em  any} determination of the log. 
Clearly, when
$a_i$ is a rational positive integer, $\log a_i$ is uniquely determined.  Furthermore, 
if $d$ is the degree of the field extension $\rat(a_1,\ldots,a_n)$
over $\rat$, then
in \cite{BakerWust93},  the
``height'' function $h'(a_i)$  is defined instead to be 
$h'(a_i) = \max \{ h(a_i), (1/d)|\log a_i| , 1/d \}$,
where $h(a_i)$ is the {\em absolute logarithmic Weil height}
of $a_i$, which we will discuss next.   
First note that in the setting of positive integers, $a_i$, clearly
$d = 1$, and thus  $h'(a_i) = \max \{ h(a_i), \log a_i , 1\}$.
Now, one way to define the absolute logarithmic Weil height $h(a_i)$ 
(see \cite{Schinzel00}) 
is this: let $p(x) \in {\mathbb Z}[x]$ be the minimal polynomial 
for the algebraic number $a_i$, suppose $p(x)$ has degree $d$,
let $f_0$ be the leading coefficient of $p(x)$,
and let $\alpha_1, \ldots,\alpha_d$ be the complex roots of $d$ (with
repetition).   Then $$h(a_i)  =  \log ( (|f_0| \prod^d_{i=1} \max(1,|\alpha_i|))^{1/d})$$
Now, in the simple case where $a_i$ is a positive integer, we see immediately that its minimal
polynomial is $p(x) \equiv  x - a_i$, and thus  that $h(a_i) = \log a_i$.
Thus, for positive integers $a_i$,  indeed $h'(a_i) =  \max \{ \log a_i, 1 \}$, as given.
}, where both lists have the same length $n$,
if $\Lambda(a,b) \neq 0$,
then 
$$ |\Lambda(a,b) |  \geq G(a,b).$$
\end{theorem}

A lower bound similar to Baker-W\"{u}stholz's was obtained by 
Waldschmidt \cite{Waldschmidt-93}.
A somewhat improved bound was obtained more recently by Matveev 
\cite{Matveev-2000I,Matveev-2000II}, who showed 
$|\Lambda(a,b) |  \geq H(a,b)$,
where $H(a,b) := e^{-C'(n) h'(a_1) h'(a_2) \ldots h'(a_n) \log B}$
and $C'(n) := 2.9 (2e)^{2n+6} (n+2)^{9/2}$.
(See Nesterenko's presentation \cite{Nesterenko-2003}.)  
The improved bound by Matveev does not have any stronger consequences
for our complexity theoretic purposes, beyond that of Theorem \ref{thm:bak-wust}.

These results constitute
the current state of the art: they are the best known 
lower bounds for (homogeneous) linear forms in
logarithms of rational numbers (and for more general numbers).  
Next we state an older conjecture of Lang and Waldschmidt, which 
would significantly strengthen both Theorem \ref{thm:bak-wust}
and Matveev's improved bound:

\begin{myconj}[Lang-Waldschmidt, 1978 (cf. \cite{LangDio78,Waldschmidt-survey04})]
\label{conj:lang-wald}
For every $\epsilon > 0$, 
there is a $C(\epsilon) > 0$, such that
for any list of positive integers $a$ and any list 
of non-zero integers $b$, where both lists $a$ and $b$ are 
of length $n$, if  $\Lambda(a,b) \neq 0$,  then:

\begin{equation}\label{eq:lang-wald}
|\Lambda(a,b) |  \geq  
\frac{C(\epsilon)^n B}{( |b_1| \ldots |b_n| |a_1| \ldots |a_n|)^{1 + \epsilon}}
\end{equation}
\end{myconj}

We next recall a central conjecture in modern number theory,
namely the {\em ABC  conjecture} (due to Masser and Oesterl\`{e}), and a more recent
refinement of the ABC conjecture, given by Baker.
See, e.g., \cite{baker-wustholz-survey,granville-abc-2002}, for background on the conjecture.
For any integer $m$,  let $N(m) := (\prod_{p | m} p)$, denote the product of all 
distinct prime numbers $p$ that divide $m$ (i.e., without repetition of any prime $p$).

\begin{myconj}[ABC conjecture \cite{Masser-abc,Oesterle-abc}]
\label{conj:abc}
For every $\epsilon > 0$, 
there is a  $K(\epsilon) > 0$,
such that
for any positive integers $a,b,c$,
such that $a+ b = c$, and such that $a,b,c$ are relatively
prime (i.e., such that $GCD(a,b) = 1$),
we have:
$$ c \leq K(\epsilon) N(abc)^{1+\epsilon}.$$ 
\end{myconj}

For any positive integer $m$, let $\omega(m)$ denote the number
of distinct prime numbers that divide $m$.
In \cite{Baker-abc-1998}, Baker put forward several refinements
of the ABC conjecture which make the ``constants'' more explicit.
Among them was the following:

\begin{myconj}[Baker's refinement of the ABC conjecture \cite{Baker-abc-1998}]
\label{conj:abc-baker}
There are absolute constants, $K, K' > 0$, 
such that for any integers $a, b, c$ 
that are relatively prime,
and such that $a + b + c  = 0$,   for any $\epsilon > 0$,
we have\footnote{It is not
obvious that Conjecture \ref{conj:abc-baker} is a refinement
of (i.e., implies) 
the standard ABC Conjecture \ref{conj:abc}, but 
as pointed out in \cite{Baker-abc-1998} this is the case,
the key reason being that 
 for integers $n > 1$, $\omega(n) \in O(\log(n)/\log(\log(n)))$.
Indeed, a little calculation shows that 
Conjecture \ref{conj:abc-baker} implies the ABC conjecture
with $K(\epsilon) \in 2^{2^{O(1/\epsilon)}}$. }:\\
\centerline{$\max (|a|,|b|,|c|) \leq  K' \epsilon^{-K\omega(ab)} N(abc)^{1+\epsilon}. $}
\end{myconj}
Baker shows in \cite{Baker-abc-1998} that Conjecture
\ref{conj:abc-baker} implies the following 
(slightly weakened) variant of the Lang-Waldschmidt conjecture:

\vspace*{0.04in}

\noindent {\em {\bf Consequence of Conjecture \ref{conj:abc-baker} (see 
\cite{Baker-abc-1998}):} There is some absolute constant $K''$  
such that
for any list of positive integers $a$ and any list 
of non-zero integers $b$, where both lists $a$ and $b$ are 
of length $n$,
if  $\Lambda(a,b) \neq 0$,  then

\begin{equation}\label{eq:baker}
|\Lambda(a,b) |  \geq 
e^{- K'' (\log a_1 + \log a_2 + \ldots + \log a_n) (\log \max_i |b_i|)}
\end{equation}

}

In fact, Baker further shows in \cite{Baker-abc-1998} that
a more general (p-adic) version of the
bound (\ref{eq:baker}) 
implies the  ABC conjecture.
Thus, as noted in \cite{granville-abc-2002},
the ABC conjecture and such 
improved quantitative bounds on linear forms in logarithms 
are intimately related questions.
It is perhaps worth mentioning that in \cite{Baker2007} 
Baker expresses doubt about the stronger Lang-Waldschmidt Conjecture. 
However, he then states
a conjecture implying bound (\ref{eq:baker}), 
which is strong enough for our purposes, and he notes that 
it originates from his refined ABC conjecture in \cite{Baker-abc-1998}.

Let us now explore the intimate connection between Problem 2,
Theorems \ref{thm:bak-wust},
and Conjectures \ref{conj:lang-wald},
 \ref{conj:abc}, and \ref{conj:abc-baker}.
Suppose we want to decide whether  $a^b \geq c^d$.  Clearly, we can first check for equality (in P-time).
If equality does not hold, then our goal is 
to decide $a^b > c^d$, knowing $a^b \neq c^d$.
Equivalently, our goal is to decide whether $a^b/c^d > 1$, given that $a^b/c^d \neq 1$.
Equivalently, we want to decide whether $\log (a^bc^{-d}) > 0$, given that $\log(a^bc^{-d}) \neq 0$.

So, Problem 2 is P-time equivalent to the following problem:
given positive integers $a = \langle a_1,\ldots,a_n \rangle$, 
and  integers $b = \langle b_1,\ldots,b_n \rangle$,  both encoded in binary, 
decide whether
$\Lambda(a,b) > 0$, with the promise that $\Lambda(a,b) \neq 0$.

We can compute an approximation
of the logarithmic form $\Lambda(a,b)$, to within any given desired additive error,
$\epsilon > 0$,  in time polynomial in the encoding size of $a$ and $b$, and in $\log(1/\epsilon)$.
To see this,
we first observe the well known fact that logarithms
of integers can be approximated in P-time (in the standard Turing model).  This is 
of course classic.
Nevertheless we provide a proof, both
for completeness and because most treatments of
the numerical computation of logarithms only consider arithmetic complexity,
rather than complexity in the Turing model.
Recall we use $\log(x)$ to denote the {\em natural} logarithm of $x$,
and use $\log_2(x)$ to denote the $\log$ base 2.

\begin{proposition}
\label{prop:approx-log}
There is an algorithm that,
given a positive integer $a$, encoded in binary, and given
a positive integer $j$, encoded in unary, computes a rational value $v_a$,
such that $$| \log(a) - v_a | < 2^{-j}$$
The algorithm runs in time polynomial
in $j$ and $\log_2(a)$  (in the Turing model).
\end{proposition}

Proposition \ref{prop:approx-log} is proved in the appendix.  We
can use it to easily prove:

\begin{proposition}
\label{prop:approx-log-form}
Given as input positive integers $a = \langle a_1,\ldots,a_n \rangle$ 
and 
integers $b = \langle b_1,\ldots,b_n \rangle$,  both encoded in
binary,  and given a positive integer $j$ (given in unary),  there is 
an algorithm that runs in P-time
(i.e., in time polynomial in $j$ and in the encoding
size of lists $a$ and $b$) and outputs rational 
value $v_{a,b}$,
such that  $|\Lambda(a,b) - v_{a,b}| < \frac{1}{2^j}$.
\end{proposition}

\begin{proof}
Given $a= \langle a_1, \ldots, a_n \rangle$, and given $b = \langle b_1, \ldots, b_n \rangle$,
for each $i=1,\ldots,n$ we will  first compute an approximation $v_{a_i}$ of $\log(a_i)$ such that

$$| \log(a_i) - v_{a_i} | < \frac{1}{2^{j + \log_2(|b_i|) + \log_2(n)}}$$

By Proposition \ref{prop:approx-log} we can compute such a $v_a$ in time polynomial
in the encoding size of the input, $a$ and $b$, and in $j$.

Then we let $v_{a,b} :=   \sum^n_{i=1} b_i v_{a_i}$.  
Finally we observe that this yields:

\begin{eqnarray*}
| \log(\Lambda(a,b)) - v_{a,b} |  & = &   | \sum^n_{i=1}  (b_i \log(a_i)) - b_i v_{a_i} | \\
& \leq & \sum^n_{i=1}  |b_i| |\log(a_i) - v_{a_i}| \\
& \leq & \sum^n_{i=1}  |b_i| 2^{-\log{|b_i|}} 2^{-j -\log_2(n)} \\
& = & \sum^n_{i=1} 2^{-j}/n  =  2^{-j}
\end{eqnarray*}

Thus the rational number $v_{a,b}$, which can be computed in time polynomial in $j$, and in the encoding
size of $a$ and $b$, is the desired additive approximation of $\Lambda(a,b)$.
 \end{proof}

\begin{proposition}\mbox{}
\label{prop:main-comp-prop}
\begin{enumerate}

\item If the ABC conjecture, as refined by Baker (Conjecture \ref{conj:abc-baker}), holds, or if the
Lang-Waldschmidt conjecture (Conjecture \ref{conj:lang-wald}) holds, then
Problem 2 is decidable in P-time.

\item (cf. Shub \cite{Shub-smalefest-93}, Theorem 6)  When $m$ and $n$ are fixed constants, 
Problem 2 is (unconditionally) decidable in P-time.\footnote{As mentioned 
in the introduction,
Shub (Theorem 6 in \cite{Shub-smalefest-93}, page 454) already observed part (2.).  
Shub sketched a proof based on precisely the same ideas as 
ours, but the sketched proof
in \cite{Shub-smalefest-93} mis-states the known lower bounds on linear forms in logarithms.
In fact, the lower bound quoted in \cite{Shub-smalefest-93} is provably 
false, and violates
Dirichlet's approximation theorem.  Namely, \cite{Shub-smalefest-93} 
states that for any positive integers $a_1,\ldots,a_m$, 
and non-zero integers $n_1,\ldots,n_m$,
if $\sum^m_{i=1} n_i \log(a_i) \neq 0$, then $| \sum^m_{i=1} n_i \log(a_i) | > J(\bar{a}) > 0$, 
where $J(\bar{a}) = 2^{-K \cdot m \cdot (\log(\max_i a_i))^{m} \cdot \log \log({\max_i a_i})}$,
for some fixed constant $K$.  Note that this stated lower bound $J(\bar{a})$ 
depends only on $\bar{a}= (a_1,\ldots,a_m)$,
and is independent of the coefficients $n_i$.
However, this is false already for linear forms in two logarithms.  
Consider, e.g., $\log(3)$ and
$\log(5)$, and let $\alpha := \log(3)/\log(5)$.  By  Dirichlet's approximation theorem, for
any real number $\alpha$, and for all $\epsilon > 0$, there is a rational number
$p/q$, such that $| \alpha -  p/q | < \epsilon/q$,  and thus that $|q  \log(3) - p \log(5) | < \epsilon \cdot \log(5)$.
Thus note that we can choose $\epsilon = \epsilon'/\log(5)  > 0$, for an arbitrary $\epsilon' > 0$.
Thus, for all $\epsilon' > 0$,  there exist integers $q$ and $p$ such that $| q \log(3) - p \log(5) | < \epsilon'$.
This contradicts the lower bound on linear forms in logarithms quoted 
by Shub in \cite{Shub-smalefest-93}.
For this reason, we provide a proof here, using the best currently known lower bounds.
}

\end{enumerate}
\end{proposition}

\begin{proof}
To prove both (1.) and (2.) we simply compute a
sufficiently good additive
approximation to $\Lambda(a,b)$, using Proposition \ref{prop:approx-log-form}, 
and we then 
apply
the ABC conjecture (Conjecture \ref{conj:abc-baker})
and its consequence (\ref{eq:baker}) or the
 Lang-Waldschmidt Conjecture (Conjecture \ref{conj:lang-wald}). 
Likewise, to  obtain (2.), 
after approximating $\Lambda(a,b)$ we apply the
Baker-W\"{u}stholz Theorem (Theorem \ref{thm:bak-wust}).  

In more detail, we
start by proving (2.): we are given lists $a$ and $b$ of length $n$ and lists $c$ and $d$ of length $m$, with
$m$ and $n$ fixed constants.
Let $A :=  \max \{ \max_i a_i,  \max_j c_j \}$.  Let $B :=  \max \{ |b_1|, \ldots, |b_n|, 
|d_1|, \ldots, |d_m|, e \}$.
First we check in P-time if $a^b = c^d$. If this is the case, then we are done.
So assume that $a^b \neq c^d$.
If $a^b c^{-d} \neq 1$, i.e., $\log(a^b c^{-d}) \neq 0$, then
by Theorem \ref{thm:bak-wust} 
we have
\begin{equation}
\label{eq:for-comp-final}
| \log(a^b c^{-d}) | \geq 2^{-K (\prod^{n}_{i=1} \log(a_i)) (\prod^m_{j=1} \log(c_j)) \log B} \geq 
2^{-K (\log A)^{m+n} \log B}
\end{equation}
for some fixed constant $K$ (that depends on $n$ and $m$).
But by Proposition \ref{prop:approx-log-form},  when $m$ and $n$ are fixed constants,
we can compute in time polynomial in the encoding size of $a$, $b$, $c$ and $d$,  a rational
value $v_{a,b,c,d}$ such that 

\begin{equation}
\label{eq:for-comp-2-final}
 | \log(a^b c^{-d}) - v_{a,b,c,d} | < 2^{-K (\log A)^{m+n} \log B} \ .
\end{equation}
Now suppose $v_{a,b,c,d} \geq 0$.   Then if $ \log(a^b c^{-d}) \leq 0$, we would have 
$| \log(a^b c^{-d}) - v_{a,b,c,d} | =   v_{a,b,c,d} +  |  \log(a^b c^{-d})| \geq 
2^{-K (\log A)^{m+n} \log B}$, the last inequality following by 
(\ref{eq:for-comp-final}).  However, this contradicts the inequality 
(\ref{eq:for-comp-2-final}) just given.
Thus if $v_{a,b,c,d} \geq 0$, then it follows that $\log(a^b c^{-d}) > 0$, i.e., 
that $a^b > c^d$.

Similarly, suppose $v_{a,b,c,d} < 0$.  Then if $\log(a^b c^{-d}) \geq 0$, we would
have  $| \log(a^b c^{-d}) - v_{a,b,c,d} | =   |v_{a,b,c,d}| +  |  \log(a^b c^{-d})| \geq
2^{-K (\log A)^{m+n} \log B}$, by inequality (\ref{eq:for-comp-final}). 
However, again, this contradicts (\ref{eq:for-comp-2-final}).
Thus if $v_{a,b,c,d}< 0$ then $\log(a^b c^{-d}) < 0$, i.e., $a^b < c^d$.

\item To prove (1.), e.g., in the
case where we use the consequences of Baker's refined ABC conjecture
(Conjecture \ref{conj:abc-baker}, and in particular the bound (\ref{eq:baker})),  
exactly the same argument goes through if we instead compute an approximation $v_{a,b,c,d}$ 
such that 

$$| \log(a^b c^{-d}) - v_{a,b,c,d} | <  2^{-K'''  ((\sum^{n}_{i=1} \log(a_i)) + (\sum^m_{j=1} \log(c_j))) (\log B)}$$
which again, we know we can do in time polynomial in the encoding size of $a$, $b$, $c$, and $d$.
Similarly, exactly the same argument goes through for proving (1.) based on the Lang-Waldschmidt conjecture.
 \end{proof}

\section{Maximum probability parsing for SCFGs}

\label{app-sec:application-scfg}

We now describe an application to the problem of computing a
{\em maximum probability parse tree} for a given string
on a given arbitrary  {\em stochastic context-free grammar} (SCFG).
For particular classes
of grammars (e.g., those in Chomsky Normal Form) 
there are well known dynamic programming algorithms for this
(based on the CKY parsing algorithm),
which generalize the well-known Viterbi algorithm
for HMMs \cite{ManSch99}. 
For arbitrary SCFGs with $\epsilon$-rules, the problem
is more involved, and there do not exist
good complexity bounds in the literature.

\medskip

\noindent {\bf Definitions and Background for SCFGs.}
An SCFG $G =  (V, \Sigma,  R, S,p)$ consists
of a finite set $V$ of {\em nonterminals}, a 
{\em start} nonterminal $S \in V$,
a finite set $\Sigma$ of  {\em alphabet (terminal)} symbols,
and a finite list of {\em rules}, $R \subset V \times (V \cup \Sigma)^*$,
where each rule $r \in R$ is a pair $(A,\gamma)$,
which we usually denote by $A \rightarrow \gamma$,
where $A \in V$ and $\gamma \in (V \cup \Sigma)^*$.
Finally $p: R \rightarrow \real^+$ maps each rule $r \in R$ to a positive
{\em probability}, $p(r) > 0$.
For computational purposes we assume as usual that the rule probabilities are
rational numbers, given as the ratios of two integers written in binary. 
We often denote a rule $r = (A \rightarrow \gamma)$  together with its probability
by writing $A \stackrel{p(r)}{\longrightarrow} \gamma$.
Note that we allow $\gamma \in (V \cup \Sigma)^*$ to possibly be 
the  empty string, denoted by $\epsilon$.
A rule 
of the form $A {\rightarrow} \epsilon$ is called an {\em $\epsilon$-rule}.
For a rule $r = (A \rightarrow \gamma)$, we 
let $\leftr(r) := A$ and $\rightr(r) := \gamma$. 
We let $R_A = \{ r \in R \mid \leftr(r) = A\}$.
For $A \in V$, let $p(A) = \sum_{r \in R_A} p(r)$.
A SCFG must satisfy that $\forall  A \in V$, $p(A)  \leq 1$.
An SCFG is called {\em proper}
if $\forall A \in V, \; p(A) = 1$.

For a SCFG, $G$, 
strings $\alpha, \beta \in (V \cup \Sigma)^*$, and 
$\pi = r_1 \ldots r_k \in R^*$, we write 
$ \alpha \stackrel{\pi}{\Rightarrow} \beta$
if the leftmost derivation starting from $\alpha$, and  
applying the sequence $\pi$ of rules, derives $\beta$.
We define the probability $p( \alpha \stackrel{\pi}{\Rightarrow} \beta)$
of the derivation to be $p( \alpha \stackrel{\pi}{\Rightarrow} \beta) = 
\prod^k_{i=1} p(r_k)$ if $\alpha \stackrel{\pi}{\Rightarrow} \beta$,
and let $p(\alpha \stackrel{\pi}{\Rightarrow} \beta) = 0$ otherwise.
If $A \stackrel{\pi}{\Rightarrow} w$ for $A \in V$ and $w \in \Sigma^*$,
we say that $\pi$ is a {\em complete} derivation from $A$
and its {\em yield} is $y(\pi) =w$.
There is a natural one-to-one correspondence between the complete (leftmost) derivations
of $w$ starting at $A$ and the {\em parse trees} of $w$
rooted at $A$, and this correspondence preserves 
probabilities.
Recall that a parse tree is a rooted, ordered finite tree,
where every leaf $v$ is labeled with a symbol $l(v) \in \Sigma \cup \{\epsilon\}$,
every internal (non-leaf) node $v$ is labeled with 
a nonterminal $l(v) \in V$ and has an associated rule 
$r(v) \in R_{l(v)}$ whose right-hand side is the
concatenation of the labels of the children of $v$.
The yield of the parse tree is the concatenation of the labels of the leaves.
The probability of the parse tree is the product of the probabilities of the
rules associated with its internal nodes.

For a non-terminal $A$ and a string $w$,
the {\em maximum} parse tree probability for $w$, starting at $A$, is defined 
to be $p^{\max}_{A,w} = \max_{\pi \in R^*}  p(A \stackrel{\pi}{\Rightarrow} w)$.
The {\em total} probability of generating $w$ starting at $A$
is defined by $p_{A,w} = \sum_{\pi \in R^*} p(A \stackrel{\pi}{\Rightarrow} w)$.
Given an SCFG, $G = (V,\Sigma,R,S,p)$, 
and given a string 
$w \in \Sigma^*$,  the goal of maximum probability
parsing is to compute $p^{\max}_{S,w}$ and also to compute (a representation of)
a maximum probability parse tree, i.e., 
$\arg \max_{\pi \in R^*}  p(S \stackrel{\pi}{\Rightarrow} w)$.
In the following we will also use $p^{\max}_{G,w}$ 
to denote the maximum probability $p^{\max}_{S,w}$
of a parse tree for $w$ from the start nonterminal $S$ of $G$.

For general SCFGs, $G$, that have $\epsilon$-rules, 
the maximum probability parse tree of a string $w$ (even just the string $w=\epsilon$)
may have exponential size in the size of the grammar, and the corresponding
maximum probability may need an exponential number of bits when written 
as the ratio of two integers.
The probability can be specified more compactly in PoE 
notation. For any SCFG $G$ and string $w$,
polynomial size (in $|G|$ and $|w|$) always suffices to
represent the maximum parsing probability in PoE notation:
the bases of the expression are (a subset of) the given rule probabilities 
of $G$ and the exponents are the number of occurrences of the rules in
the optimal parse tree;
the numbers of occurrences are at most exponential and thus can be written 
in a polynomial number of bits.

The optimal parse tree can be specified more compactly using a DAG
representation that identifies isomorphic subtrees.
Formally, a {\em parse DAG} is a rooted, ordered DAG $D$ (i.e. the DAG
has a single source, and the children of
every node are ordered), where every sink (leaf) node is
labeled with a symbol $l(v) \in \Sigma \cup \{\epsilon\}$,
every other (non-sink) node $v$ is labeled with 
a nonterminal $l(v) \in V$ and has an associated rule 
$r(v) \in R_{l(v)}$ whose right-hand side is the
concatenation of the labels of the children of $v$.
For every node $v$ we can define inductively in a bottom-up order the
yield and probability of the subDAG $D[v]$ rooted at $v$: 
If $v$ is a leaf then the yield is $l(v)$ and the probability is 1.
If $v$ is an internal node, then the yield of $D[v]$ is the concatenation
of the yields of the children of $v$, and the probability of $D[v]$
is the product of the probability of the rule $r(v)$ and the
probabilities of the subDAGs rooted at the children of $v$.
The parse DAG $D$ corresponds to a parse tree $T$ obtained by 
replicating nodes so that every node other than the root has a unique parent;
the yield and probability of the DAG are equal to the yield and probability of the corresponding tree.
As we shall see, for every SCFG $G$ and string $w$ in the language $L(G)$ of $G$
(i.e., that has non-zero probability), there is a maximum probability 
parse DAG for $w$ of size polynomial in $|G|$ and $|w|$.\footnote{The succinctness of the DAG representation is essential only for derivations of $\epsilon$
from nonterminals. In particular, for every SCFG $G$ and $w \in L(G)$,
there is a maximum probability parse DAG for $w$,
of size polynomial in $|G|$ and $|w|$, which consists of a tree, 
some of whose leaves are replaced by DAGs with yield $\epsilon$.}
Our goal is to construct such a maximum probability parse DAG.

We will say that an SCFG, $G=(V,\Sigma,R,S,p)$ is in 
{\em Simple Normal Form} (SNF) if every nonterminal $A \in V$ belongs to one of the following three types:
\vspace*{-0.06in}
\begin{enumerate}
\item type {\tt L}: every rule $r \in R_A$,  has the form $A \xrightarrow{p(r)} B$, for some $B \in V$.
\item type {\tt Q}: there is a single rule in $R_A$: $A \xrightarrow{1} BC$, 
for some $B, C \in V$.
\item type {\tt T}: there is a single rule in $R_A$:  either 
$A  \xrightarrow{1} \epsilon$, 
or $A  \xrightarrow{1} a$ for some $a \in \Sigma$.
\end{enumerate}

An SCFG is said to be in {\em Chomsky Normal Form} (CNF) if it satisfies 
the following conditions:
\begin{itemize}
\item  The grammar does not contain any $\epsilon$-rule 
except possibly for a rule $S \stackrel{p}{\rightarrow}
\epsilon$
associated with the start nonterminal $S$;   if it does contain 
such a rule, then $S$ does not
appear on the right hand side of any rule in the grammar.

\item Every rule, other than $S \stackrel{p}{\rightarrow} \epsilon$,
is either of the form $A \stackrel{p}{\rightarrow} BC$,
or of the form $A \stackrel{p}{\rightarrow} a$
where $A$, $B$, and $C$ are nonterminals in $V$ and $a \in \Sigma$ is
a terminal symbol.
\end{itemize}

We shall show below that every SCFG can be converted efficiently
in P-time, to an equivalent SCFG that is in SNF form,
where the equivalence also entails a 
probability-preserving bijection between parse trees of strings
in the two grammars.

Unlike SNF form, 
conversion of even an ordinary context-free grammar 
to CNF form does not in general preserve 
a bijection between parse trees of the original grammar
and those of the CNF grammar. 
This is so even when we ignore the additional issue of 
needing to preserve the probability
of corresponding (e.g., maximum probability) parse trees for a given string, 
in the setting of SCFGs. 
 
Furthermore, as shown in \cite{ESY12}, it is not possible in general
to convert an arbitrary SCFG to one in CNF form which preserves
even just the probabilities of generating every terminal string, without having
to introduce irrational rule probabilities even when all
rule probabilities of the original SCFG were rational.
See \cite{ESY12} for a considerably 
involved P-time 
algorithm that converts an SCFG to an {\em approximately} equivalent
CNF form SCFG. Here ``approximately equivalent'' only refers
to preservation of the probability of generating strings up to a given length,
and not to preservation of a correspondence between parse trees
(which is not doable in general),  and thus such a conversion
is not suitable when our goal is to compute, e.g.,  a maximum probability
parse tree for a given string on a given SCFG.

\begin{lemma}(Lemma B.5  of \cite{ESY12})
\label{grammar-sd2nf-lem}
Any SCFG, $G = (V,\Sigma,S,R,p)$, with rational rule probabilities, 
can be converted in linear time to a SCFG, $G' = (V',\Sigma,S,R',p')$, 
in SNF form, such that $V \subseteq V'$,
such that $G'$ has the same set of rational values as rule probabilities
(possibly with some additional rules having probability 1),
and such that for every nonterminal $A \in V$ and string $w \in \Sigma^*$, 
there
is a probability preserving
bijection between parse trees of $w$ rooted at $A$ in $G$
and parse trees of $w$ rooted at $A$ in $G'$.
Moreover, given a parse tree for $w$ rooted at $A$ in $G$
we can easily recover the corresponding parse tree in $G'$,
and vice versa.

\end{lemma}

The proof is directly analogous to  
related proofs in \cite{ESY12},  
and simply
involves adding new auxiliary nonterminals  and 
new auxiliary rules, each  having probability 1,  
to suitably ``abbreviate'' the sequences of symbols, $\gamma$,
that appear 
on the right hand side (RHS) of rules 
$A \stackrel{p}{\rightarrow} \gamma$, whenever $|\gamma| \geq 3$.
We do this repeatedly until for all such RHSs, $\gamma$, we 
have $|\gamma| \leq 2$.   
To obtain the normal form,
we may then also need to introduce nonterminals that generate 
a single terminal symbol  with probability 1.
The resulting SCFG is guaranteed to
be at most polynomially larger (in fact, only linearly larger)
if we are careful.

We give an efficient algorithm {\em operating in
the unit-cost arithmetic RAM model with only multiplication ($\{ * \}$) operations and comparisons permitted}, for computing the
maximum likelihood parse tree of a given string.

\begin{theorem}
\label{thm:main-unit-cost-max-alg}
Given any SCFG, $G$, in SNF form,
with rational rule probabilities,\footnote{We can even allow 
irrational rule probabilities, since for this
we assume a unit-cost RAM model of computation.}
given a string $w \in \Sigma^*$, 
there is an algorithm that computes the maximum parse tree probability
$p^{\max}_{G,w}$, and
if $p^{\max}_{G,w} > 0$, then it also  
computes a succinct DAG representation of 
a maximum probability parse tree, $t^{\max}_w$, for $w$.

The algorithm runs in polynomial time {\em in the unit-cost
arithmetic RAM model of computation} 
{\em with \underline{only} multiplication operations (and comparisons) allowed.}
And thus, it is P-time Turing reducible to Problem 2: inequality
comparison of
integers in succinct PoE representation.
\end{theorem}

\begin{proof}
Given the SCFG, $G$, we first compute, for every non-terminal $A$,
the maximum probability $p^{\max}_{A,\epsilon}$ of any finite parse tree
that starts at nonterminal $A$ and generates the empty string $\epsilon$.

We do this using a variant of Dijkstra's shortest path algorithm,
due to Knuth \cite{Knuth77}, 
which works not on finite graphs but on weighted context-free grammars,
for generating a parse tree with smallest {\em sum} total weight
(as well as other classes of weight functions).
See the survey on probabilistic parsing 
by Nederhof and Satta \cite{NedSat08b} (their Figure 5) where
they nicely describe Knuth's
algorithm applied to the specific problem of computing, for any given SCFG,
the maximum probability of any finite
parse tree.
We follow  Nederhof and Satta's 
description (\cite{NedSat08b}).

Given the original SCFG, $G$, and given a nonterminal $A$,
if we are interested in computing $p^{\max}_{A,\epsilon}$, we  
first remove from $G$ all rules of the form $B \rightarrow \gamma$,
for any nonterminal $B$, and string $\gamma$
where $\gamma$ contains at least one terminal symbol 
$\sigma \in \Sigma$ in it, because such rules can't possibly help us
generate the empty string $\epsilon$.

Let us call the resulting SCFG after these removals $G'$.\footnote{Some
non-terminals $A_i$ in $G'$ may now not have a set of rules $R_i$
associated with them
whose probabilities sum to 1, because
we have removed some rules.  This causes no problem in 
our computations: we are interested in probabilities of parse trees
that don't involve the removed rules, and these remain the same as
in $G$.}

Every finite parse tree of the remaining SCFG, $G'$,
must generate the empty string, because no other terminal symbols
are left.  Moreover, there is a one-to-one probability
preserving correspondence
between parse trees of $G$ generating $\epsilon$ and parse
trees of $G'$,  starting at any nonterminal $A$ in the two
SCFGs, respectively. 

Therefore, computing $p^{\max}_{A,\epsilon}$ is equivalent to computing
the maximum probability of {\em any} finite parse tree starting
at $A$ in $G'$.  Let us denote this probability by 
$p^A_{G',\max}$.

Knuth's variant of Dijkstra's algorithm is 
described in Figure \ref{fig:knuth-alg} which is
taken from \cite{NedSat08b} 
(their Figure 5).\footnote{\label{footnote:log-transform} We note, for clarity, that ``Dijkstra's algorithm'' is usually viewed as 
a single-source {\em shortest path} algorithm on a edge-weighted directed graph,
i.e., an algorithm that finds a path from a given source $s$ to
targets $t$ which {\em minimizes} the {\em sum} of the 
non-negative edge weights along the path.
And Knuth's variant can also be viewed as computing the minimum 
{\em sum} total weight of all
rules in a finite parse tree starting from a given non-terminal.
However,   
a well-known and straight-forward transformation
shows that Dijkstra's algorithm (and Knuth's algorithm) can also be used to compute a {\em maximum}
probability path from a source $s$ to a target $t$
(respectively, a maximum
probability parse tree starting at a given non-terminal).  Namely, maximizing
the {\em product} of probabilities labeling edges along 
a path from $s$ to $t$ is equivalent to minimizing the length of a path
from $s$ to $t$, when every edge having probability $p > 0$ 
in the original graph
is assigned the non-negative weight $- \log p$ in the transformed graph.
And the same transformation also works for Knuth's variant
of Dijkstra's algorithm for weighted CFGs, given
in Figure \ref{fig:knuth-alg}.
For establishing the polynomial running time  
of these algorithm in the {\em unit-cost} {\em (exact) arithmetic} model
of computation, it is convenient  
to view Knuth's (and Dijkstra's) algorithm
in their multiplicative form, because this avoids the need to consider
approximations of $-\log p$, for given rational rule probabilities
$p$.   However, when we operate in the standard Turing model of computation,
as in the proof of Corollary \ref{approx} for {\em approximating}
in P-time the maximum probability of a 
parse tree,  we indeed use the $\log$-transformed (minimization)
variants of these same algorithms, by first approximating the 
non-negative edge weights $-\log p$.}
It  computes $p^B_{G',\max} $
for every nonterminal $B$ of a given SCFG $G'$,
until it has computed $p^A_{G',\max}$, 
or else discovered that $p^A_{G',\max} = 0$.
It does so by iteratively finding a 
non-terminal $B$ for which $p^{B}_{G',\max}$ has not yet been computed,
and such that, among all such nonterminals, $B$ gives rise to the maximum
product probability of a rule $B \rightarrow \gamma$
times the maximum parse tree probabilities $p^{B'}_{G',\max}$
for all nonterminal occurrences $B'$ in $\gamma$.

\begin{figure}

\noindent \begin{algorithmic}

\State $D :=  \Sigma \cup \{ \epsilon \}$;
\State Initialize: $p^A_{G',\max} := 0$, for every nonterminal $A$;
\State  Define $p^{a}_{G',\max} := 1$ for all $a \in D$.
\While{($F :=  \{ B \mid  B \not\in D \ \wedge \ \exists r = 
(B \rightarrow X_1 \ldots X_m)  \ \mbox{where} \ X_1, \ldots, X_m\in D \} 
\neq \emptyset$)}

\ForAll{$B \in F$} 
\State  $q(B) := \max_{\{r = (B \rightarrow X_1 \ldots X_k) \mid
X_1, \ldots, X_k \in D\}} p(r)  
p^{X_1}_{G',\max} \ldots p^{X_k}_{G',\max}$
\State $m(B) := \arg\max_{\{r = (B \rightarrow X_1 \ldots X_k) \mid
X_1, \ldots, X_k \in D\}} p(r)  
p^{X_1}_{G',\max} \ldots p^{X_k}_{G',\max}$.
\EndFor
\State Let $C := \arg\max_{C \in F} q(C)$;
\State $p^{C}_{G',\max} :=  q(C)$;
\State  $D := D \cup \{ C \}$;
\State
       $t^C_{G',\max} := m(C)$; 
\State (where $m(C)$ only describes the root rule
of tree $t^C_{G',\max}$, and 
$t^C_{G',\max}$ can thus be 
\State \ \ viewed as being encoded succinctly as a straight-line program,
i.e., a DAG.)
\EndWhile
\State 
For every nonterminal $A$:  output  $p^A_{G',\max}$, and 
if $p^A_{G',\max} > 0$ then also output $t^A_{G',\max}$;

\end{algorithmic}

\caption{\label{fig:knuth-alg} Knuth's variant of Dijkstra's algorithm, for 
computing a globally maximum probability
parse tree 
starting from every
nonterminal $A$  (see Nederhof \& Satta \cite{NedSat08b}).}
\end{figure}

It is not difficult to check that 
this algorithm works correctly, for the
same reason that Dijkstra's algorithm works correctly.
It computes $p^{A}_{G',\max} = p^{\max}_{A,\epsilon}$ 
and furthermore also computes a DAG
representation (straight-line program representation),
of a maximum probability parse tree $t^A_{G,\max,\epsilon}$ starting
at $A$,
for every nonterminal
 $A$ of both $G'$ and $G$.
(Note that this algorithm does not require the SCFG $G$ to 
be in SNF form.  We will exploit SNF form only later, in the
final dynamic programming step of the algorithm.) 

Note, furthermore, that the 
algorithm requires at most a polynomial number of arithmetic
operations, specifically, it requires 
\underline{\em multiplications only} (and this 
fact will be important for us later), as well as comparison operations 
(for allowing us to take the maximum over a finite set of values).

Thus, the algorithm clearly runs in polynomial time in the
unit-cost arithmetic RAM model of computation.

However, note that the probabilities $p^{\max}_{A,\epsilon}$ can 
clearly be extremely small positive numbers in the worst case (as small
as $1/2^{2^{|G|}}$), and thus we can not encode them in standard binary
notation in polynomial size.
However, since only multiplications are being used over a set of
input rule probabilities, $p(r)$,  we can encode these probabilities
succinctly in PoE, by specifying (in binary) the non-negative
integer power of each $p(r)$.
We can compute these numbers with the same
number of arithmetic operations over PoE notation, as long as we can compare 
PoE numbers 
efficiently.

Having computed $p^{\max}_{A,\epsilon}$ for every nonterminal $A$ of $G$,
we next want to use these quantities to compute $p^{\max}_{A,w}$ for
an arbitrary string $w \in \Sigma^*$.

In the next step, we will use another application of Dijkstra's
algorithm, this time on a finite graph, to compute,
for every pair of nonterminals $A,B$,
the maximum probability, $p^A_{\max,B}$, that 
a derivation of $G$ starting at nonterminal $A$ eventually yields the single 
nonterminal symbol $B$ as its yield.
We do this as follows:  
construct an edge-labeled directed graph, $H = (V, E)$,
with edges labeled by probabilities, 
with the nonterminals $V$ of $G$ as its nodes,
and such that, for 
each rule $A \stackrel{p}{\rightarrow} B$,
we have the corresponding edge $(A,p,B) \in E$.
For each rule $A \stackrel{p(r)}{\rightarrow} B C$ such that 
$p^{\max}_{C, \epsilon} > 0$,
we place the edge $(A,p',B) \in E$,
where $p' = p(r) * p^{\max}_{C,\epsilon}$. 
Similarly if $p^{\max}_{B,\epsilon} > 0$, 
then we put $(A,p',C) \in E$ where
$p' = p(r) * p^{\max}_{B,\epsilon}$.
For every pair of nodes $A,B$, if there are multiple edges 
$(A,p,B)$ and $(A,p',B)$ in $E$, we only
keep one edge with the maximum probability.

We then 
run Dijkstra's algorithm from every node 
$A$ to compute the maximum probability path
from $A$ to every other node $B$ in $H$,
and we let $p^A_{\max,B}$ denote the product of the probabilities
along that path.   
Note that in this case again, Dijkstra's algorithm only
requires polynomially many {\em multiplication operations}
(no additions) and comparisons, and thus runs in P-time
in the unit-cost RAM model of computation, whatever
the encoding of the probabilities labeling $H$ 
is.
While running Dijkstra's algorithm we can also retain
the maximum probability paths themselves, and combine
them with the already computed maximum probability
parse trees $t^{B}_{\max,\epsilon}$ encountered along the
path, in order to compute a DAG representation of a maximum 
probability parse tree $t^{\max}_{A,B}$ in $G$ with
root $A$ and with ``yield'' $B$.
The following claim is not difficult to prove:

\noindent{\bf Claim:} {\em This algorithm correctly
computes  the maximum probability
$p^A_{\max,B}$  of a parse tree in $G$ with root $A$ and yield $B$,
and also computes a succinct DAG representation $t^{\max}_{A,B}$ of 
such a parse tree.}

We are finally ready to iteratively compute the probabilities
$p^{\max}_{A,w}$ for an arbitrary string $w \in \Sigma^+$.
Let $w= w_1 \ldots w_n$ be the given string, with $w_i \in \Sigma$.
For $i \in \{1,\ldots,n\}$, and $j \in \{1,\ldots, n-i+1\}$,
let us define $p^{A}_{\max,i,j}$ to be the maximum probability
of any parse tree rooted in $A$ which generates the string 
$w_i \ldots w_{i+j-1}$.  
We shall now see how to compute these
$p^{A}_{\max,i,j}$'s via dynamic programming. 

Recall that we are assuming that $G$ is in SNF form.
For $j = 1$, let $q^{A}_{\max,i,1} :=  
\max_{r= (A \stackrel{p(r)}{\rightarrow} w_i)} p(r)$.
In other words, $q^A_{\max,i,1}$ is simply the maximum probability
of any rule associated with $A$ that immediately generates 
the terminal symbol $w_i$.
By default, if there is no such rule then $q^A_{\max,i,1} :=0$.
For $j > 1$,
let $q^{A}_{\max,i,j}$ denote the maximum probability
of a parse tree rooted in $A$ which generates the string
$w_i \ldots w_{i+j-1}$, {\em and furthermore} where the 
rule used at the root of the parse tree is of the form
$A \rightarrow B C$, where the yield of {\em both} the children
$B$ and $C$ are not $\epsilon$.

Figure  \ref{fig:max-dyn-prog} describes a dynamic programming
algorithm for computing both $p^A_{\max,i,j}$'s and $q^{A}_{\max,i,j}$'s
for all $i$ and $j$,
based on mutual recurrence relations.
The algorithm assumes that for all
nonterminals $A$ and $B$ the values $p^A_{\max,\epsilon}$ and
$p^A_{\max,B}$ have already been computed, as described
in the previous steps.

\begin{figure}
\begin{algorithmic}
\State {\em Initialization:}
\ForAll{nonterminals $A \in V$ and $i \in \{1,\ldots,n\}$}
\State $q^{A}_{\max,i,1} :=  
\max_{r= (A \stackrel{p(r)}{\rightarrow} w_i)} p(r)$;
\EndFor
\ForAll{nonterminals $A \in V$, and $i \in \{1,\ldots,n\}$}
\State $p^A_{\max,i,1} := \max_{B \in V}  p^{A}_{\max,B} \cdot q^B_{\max,i,1}$;
\EndFor
\State  {\em Main Loop:}
\For{$j :=  2, \ldots, n$}
\ForAll{nonterminals $A \in V$ and $i \in \{1,\ldots,n\}$}
\If{($i+j-1 \leq n$)}
\State $q^A_{\max,i,j} =
            \max_{\{m \in \{1,\ldots,j-1\} \ \mbox{\& rule}  \ r = (A \rightarrow B C) \}}
                p(r) \cdot  p^{B}_{\max,i,m} \cdot p^{C}_{\max,i+m,j-m}$
\EndIf
\EndFor
\ForAll{nonterminals $A \in V$ and $i \in \{1,\ldots,n\}$}
\If{($i+j-1 \leq n$)}
\State $p^A_{\max,i,j} = \max_{B \in V}   p^A_{\max,B} \cdot q^B_{\max,i,j}$
\EndIf
\EndFor
\EndFor
\end{algorithmic}
\caption{\label{fig:max-dyn-prog} Dynamic program
for computing maximum probability parse of given string.}
\end{figure}

It is not difficult to show that the algorithm is correct.
Note that $p^{A}_{\max,1,n} = p^A_{\max,w}$, and thus
the algorithm computes the maximum probability of a parse tree
for a string $w$. 

Note furthermore that the algorithm runs in polynomially
many steps, only requiring unit-cost {\em multiplication} operations
and comparison operations.
Furthermore, we can easily augment the algorithm so that, 
in addition to just computing $p^A_{\max,w}$ it also
computes a DAG (a straight-line program) representation
of a maximum probability parse tree $t^A_{\max,w}$ for $w$ rooted at $A$,
while still requiring only polynomially many operations in total.
This completes the proof.

It is worth mentioning that one could alternatively
use Knuth's algorithm in a slightly different way
in order to compute 
the maximum probability $p^{\max}_{G,w}$ of a parse tree
for a string $w$, in polynomial time in the unit-cost model.
Namely, we can view the string $w$, where $|w|=n$ as being described by the
obvious ``linear'' deterministic finite automaton,  $D_w$, 
with $n+1$ states, start state $s_0$, and final state $s_n$,
such that $D_w$  accepts just the single string $w$, i.e., $L(D_w) = \{ w \}$.
Then, given the SCFG, $G$, which we can assume wlog
is already in SNF form, and given the DFA, $D_w$,  we can 
use a standard ``product/intersection'' construction (see, e.g., \cite{NedSat08b}) 
to form
a new {\em weighted} CFG, $G'$, which has size polynomial in $G$ and $D_w$.
The non-terminals of $G'$ are given by triples of the form $(s,A,s')$,
where $s$ and $s'$ are states of $D_w$, and $A$ is a nonterminal of $G$.
The rules of $G'$ are formed as follows: for
every rule $A \stackrel{p}{\rightarrow} B C$ in $G$, 
we add the following rules to $G'$:  $(s,A,s') \stackrel{p}{\rightarrow}
(s,B,s'') (s'',C,s')$, for every state $s''$ of $D_w$.
For every rule of the form $A \stackrel{p}{\rightarrow} B$,
we add the rule  $(s,A,s') \stackrel{p}{\rightarrow} (s,B,s')$. 
For every rule of the form $A \stackrel{p}{\rightarrow}
a$ for some terminal symbol $a$,  we add the rule $(s,A,s') 
\stackrel{p}{\rightarrow} a$ to $G'$ if and only if 
there is a transition $(s,a,s')$ in $D_w$.
We also need to be careful with handling $\epsilon$-rules.
If $A \stackrel{p}{\rightarrow} \epsilon$ is a rule of $G$,
then we make $(s,A,s) \stackrel{p}{\rightarrow} \epsilon$ a rule of
$G'$ for every state $s$ of $D_w$.
This standard product construction has the property that,
for every non-terminal $A$ of $G$, there is a easily computable
weight-preserving (i.e., probability-preserving) 
one-to-one correspondence between the 
finite parse trees of $G$ rooted at $A$ which generate the string $w$,
and all the finite parse trees of $G'$ rooted at $(s_0,A,s_n)$.
Thus, in order to compute $p^{\max}_{G,w}$, 
we can alternatively apply Knuth's algorithm directly to $G'$ 
in order to compute the maximum probability of any finite parse tree rooted at 
$(s_0,A,s_n)$.
 \end{proof}

\begin{corollary}
Given any SCFG, $G$, with rational rule probabilities,
and given a string, $w \in \Sigma^*$, where $\Sigma$ is the terminal alphabet
of $G$:  
\begin{itemize}
\item[A.] If  either the Lang-Waldschmidt Conjecture,
Conjecture \ref{conj:lang-wald}, or 
Baker's  version of the ABC conjecture, Conjecture 
\ref{conj:abc-baker}, hold,

\item[B.] or else, if the number of distinct probabilities 
labeling the rules of $G$
is bounded by a fixed constant, $c$,
\end{itemize}
then the following all hold:

\begin{enumerate}
\item 
There is a P-time algorithm,
in the standard Turing model of computation, for computing the exact
probability $p^{\max}_{G,w}$ 
in
{\em succinct product of exponentials notation} (PoE),
and there is a P-time algorithm for 
computing, if $p^{\max}_{G,w} > 0$, 
a maximum probability parse tree $t^{\max}_w$ 
where $t^{\max}_w$ is represented succinctly as a DAG 
(straight-line program).

\item Given, additionally, a  rational probability $q \in (0,1]$,
encoded in binary\footnote{If we are assuming (A.), then 
$q$ can even be given in PoE. If we are instead assuming (B.), then 
$q$ can also be given in PoE by $a^b$, but with $a = \langle a_1,\ldots,a_k\rangle$, where $k \leq c'$, for some fixed constant $c'$.},
there is a P-time algorithm in the standard Turing model of computation
that  decides whether  $p^{\max}_{G,w} \geq q$. 

\item Likewise, given additionally another string $w' \in \Sigma^*$, 
there is a P-time algorithm (in the Turing model), that
decides whether
$p^{\max}_{G,w} \geq p^{\max}_{G,w'}$.
\end{enumerate}
\end{corollary}
\begin{proof}
The claims follow easily from the proof of
Theorem \ref{thm:main-unit-cost-max-alg}.
Specifically, recall that the unit-cost RAM
algorithms in the proof of 
Theorem \ref{thm:main-unit-cost-max-alg} only
involve multiplication of rational numbers starting with base numbers
that are rule probabilities of the given SCFG $G$, as well
as taking the maximum over such numbers 
(which we can carry out by comparisons).
We can therefore maintain all the computed numbers
in PoE format, and  based on the postulated 
conditions  it follows from
Proposition \ref{prop:main-comp-prop} 
that we can carry out all the necessary comparisons
in P-time, yielding a P-time algorithm overall
which computes the relevant probabilities in PoE notation,
and which can compare two such probabilities.
 \end{proof}

Furthermore, without any conjectures,  essentially the same algorithm 
can be used to approximate the maximum parsing probability
of a string:

\begin{corollary}
\label{approx}
Given any SCFG, $G$, with rational rule probabilities,
given a string, $w \in \Sigma^*$, where $\Sigma$ is the terminal alphabet
of $G$, and given rational $\epsilon >0$,
there is an algorithm (in the standard Turing model) 
that runs in time polynomial in $|G|$, $|w|$ and $\log(1/\epsilon)$,
which determines whether $w \in L(G)$ and if so,
computes a value $v$ such that
$| \log_2 (p^{\max}_{G,w}) -v| \leq \epsilon$
and the DAG representation of a parse tree of $w$ that has probability
$\geq (1-\epsilon) p^{\max}_{G,w}$.
\end{corollary}

\begin{proof}
We will  use essentially the same 
algorithm as in the proof of Theorem \ref{thm:main-unit-cost-max-alg},
but we will instead use the log-transformed (shortest path) variants
of Dijkstra's and Knuth's algorithms 
(see footnote \ref{footnote:log-transform}), 
by first approximating the weights $-\log p$ corresponding to 
rule probabilities $p$,
to sufficient precision.  We will show that approximating the weights 
$-\log p$ to within additive error $2^{-k}$, where $k$ is polynomial in  $|G|$,
$|w|$ and $\log(1/\epsilon)$, 
will suffice to allow the algorithm to 
approximate $p^{\max}_{G,w}$ to within the desired additive error $\epsilon > 0$.

Assume, wlog, that we first put the SCFG in SNF form.
Let $n$ be the number of nonterminals of the SNF grammar. 
Let us first estimate the size of the PoE expressions
for the probabilities that are computed by the algorithm of
Theorem \ref{thm:main-unit-cost-max-alg}.
It is easy to show that the maximum probability $p^{\max}_{A,\epsilon}$ 
of derivation
of $\epsilon$ from a nonterminal $A$ is given by a PoE expression
whose bases are rule probabilities and where the sum of the exponents 
is less than $2^{n}$.
This can be shown by an easy induction on the iterations of Knuth's algorithm
in Fig. \ref{fig:knuth-alg}, where the inductive claim is that
the sum of the exponents for a nonterminal that 
is added to $D$ in the $i$th iteration is at most $2^i-1$.

Then we construct a directed graph $H$ and use Dijkstra's algorithm
to compute probabilities $p^{\max}_{A,B}$.
The edges of $H$ have probabilities of the form
$p(r) \cdot p^{\max}_{C,\epsilon}$, and the optimal path between 
any two nodes is simple, thus it has length at most $n-1$.
Therefore, each probability $p^{\max}_{A,B}$ in PoE notation has sum of exponents
at most $(n-1)\cdot 2^n$.

If we consider then the algorithm of Fig. \ref{fig:max-dyn-prog},
it is easy to show by induction on $j$
that a probability $p^{A}_{\max,i,j}$ in PoE notation
has sum of exponents at most $(2j-1)((n-1) 2^n +1)$.
Thus, the PoE expression for the final probabilities, as well as
all the probabilities during the computation,
have sum of exponents less than $2n^2 2^n$.
Or in other words, the logarithms of the computed probabilities are
(positive) linear combinations of the logarithms of the rule probabilities
where the sum of the coefficients is less than $2n^2 2^n$.

Our algorithm for the approximate computation of the maximum
parsing probability starts by computing approximately the
logarithms of the rule probabilities to a precision
of $k = \lceil 2n + \log(1/\epsilon) \rceil$ bits, i.e. within additive error
$2^{-k} < \epsilon / 2^{2n} < \epsilon / (4n^2 2^n)$.
Then we apply the Algorithm of Theorem \ref{thm:main-unit-cost-max-alg}
using the log-transformed (additive) 
versions of Knuth's and Dijkstra's algorithms, and
doing all the computations (exactly) in logarithms, by
using the {\em approximated} values for the logarithms,
$-\log p$, of the rule probabilities $p$.
Note that every computed quantity is then a linear combination of 
the (approximated) logarithms of rule probabilities.

We now observe that the cumulative additive error that can
be introduced into the final result,
because of the initial approximation
to the logarithms of the rule probabilities, is at most 
$4n^2 2^n \cdot 2^{-k} < \epsilon$.
The reason why this holds is the following:

Consider a solution $S$, i.e., a succinctly represented
parse tree (respresented as a DAG).
Let $v(S)$ denote the logarithm of the value of this solution.
In other words, $v(S)$ denotes the log of the PoE expression 
giving the product of all rule probabilities used in $S$.
Thus, $v(S)$ can be written as $\sum_{r}  n_{r} \log p(r)$,
where the sum is over all rules $r$.
Recall that $p(r)$ denotes the probability of rule $r$.   

Assume that for all the logarithms of rule proabilities, $\log p(r)$,
we have computed an approximation, $a(r)$, 
such that $|(\log p(r)) - a(r) | < 2^{-k}$.  
Let $v'(S)$ denote the approximate log value of the solution $S$,
i.e., $v'(S) = \sum_{r}  n_{r} \cdot a(r)$.

Let $S^*$ denote the solution (parse tree) that is computed by 
the algorithm that uses the approximated values $\log p(r)$.
Let $S^{opt}$ denote the optimal solution (which would be computed
if we instead had used exact arithmetic and comparisons on PoE numbers).

First, note that $v'(S^*) \geq v'(S^{opt})$, because the 
approximation algorithm is guaranteed
to output the optimal (maximum value) solution $S^*$ with respect the
{\em approximated} logarithms $a(p(r))$ of the rule probabilities $p(r)$.

Next, note that both $v(S^*)$ and $v(S^{opt})$ 
are positive linear combinations of the actual logs of rule probabilities,
with their coefficients summing to at most $m = 2n^2 2^n$.
So $v'(S^*)$ differs from $v(S^*)$ 
by at most $2^{-k} m$, and likewise
$v'(S^{opt})$ differs from $v(S^{opt})$ by at most $2^{-k} m$.
Therefore, since we already argued that $v'(S^*) \geq v'(S^{opt})$,
it must be the case that
$v(S^*) \geq v(S^{opt}) - (2 \times 2^{-k} m)$.
Note that we have chosen $k$ such that
$2 \times 2^{-k}m = 2^{-k} 4 n^2 2^n < \epsilon$.
This completes the proof.

The algorithm computes at the same time the DAG representation 
of a parse tree $T$, whose probability satisfies
$\log_2 (p(T)) \geq \log_2 p^{\max}_{G,w} -\epsilon$.
Thus, $p(T) \geq   p^{\max}_{G,w} /2^{\epsilon} \geq (1-\epsilon)p^{\max}_{G,w}$.
 \end{proof}

\vspace*{0.05in}

\noindent {\bf Acknowledgements.}  Thanks to Howard Karloff and
Igor Shparlinski for comments on an earlier draft. In particular,
thanks to Igor for pointing out references \cite{BachDS93,bernstein05} on
efficient factor refinement.   We also thank an anonymous referee
for pointing out Theorem 6 in Shub's paper \cite{Shub-smalefest-93} to us.

\newpage

\appendix

\section{Appendix}

\subsection{Proof of Proposition \ref{prop:spe-is-same-as-arith-circ}}

\noindent {\bf Proposition \ref{prop:spe-is-same-as-arith-circ}.} {\em   
There is a simple P-time
translation from a given number represented in PoE to
the same number represented as an arithmetic 
circuit over $\{*,/\}$ with integer inputs (represented in binary).
Likewise, there is a simple P-time translation in the other direction.}

\vspace*{0.05in}

\begin{proof}
Given a number in PoE, observe first that,
for integers $a_i$ and $b_i$, we can use Horner's
rule to represent $a_i^{b_i}$  by an arithmetic circuit which takes
$a_i$ as an input, and performs repeated squaring and multiplication
by $a_i$ to obtain a circuit evaluating to $a_i^{b_i}$ whose depth and number of gates is bounded by
$\log_2(b_i)$.    We can then obviously compute a linear size circuit yielding
the product  $a^b = \prod_i a_i^{b_i}$.

In the other direction, given an arithmetic circuit $C$ over
$\{*,/\}$, with positive integer inputs $a_1, \ldots, a_n$,
we show by induction on the depth of the circuit 
that the numbers computed by a gate at depth $d$ can be translated
to a PoE $a^b$, where the base numbers are the $a_i$'s and where the exponents
$b_i$ have absolute value at most $2^d$, and 
can thus be computed in linear time, given $C$
as input.
The claim is obvious in the base case, for input gates of the 
circuit $C$, which are the $a_i$'s.  Inductively, assume that for some
gate $g_i = g_j \odot g_k$ we already have
PoE representations for $g_j$ and $g_k$, given by $a^{b(j)}$ and
$a^{b(k)}$ respectively,
where $\odot \in \{ *, / \}$.  
Assume that $d= \max \{depth(g_i), depth(g_k)\}$. By induction
we can assume each component of the lists (vectors) $b(j)$ and $b(k)$
is at most $2^d$.

Now if $\odot \equiv "*"$, then clearly the
value computed by gate $g_i$ can be represented
by $a^{b(i)}$, where $b(i) = b(j)+b(k)$ denotes component-wise addition
of the two vectors of integers $b(j)$ and $b(k)$.
Likewise, if $\odot \equiv "/"$, then clearly $g_i$ can be represented
by $a^{b(i)}$ where $b(i) = b(j) - b(k)$.   In both cases, for every component
$m \in \{1,\ldots,n\}$, we have $|b(i)_m| \leq 2^{d+1}$, since we at
most double the absolute value by adding or subtracting two numbers
with absolute value $\leq 2^d$.  Since $g_i$ has depth $d+1$, we are done.
 \end{proof}

\subsection{Proof of Proposition \ref{prop:approx-log}}

Recall we use $\log(x)$ to denote the {\em natural} logarithm of $x$,
and use $\log_2(x)$ to denote the $\log$ base 2.

\vspace*{0.09in}

\noindent {\bf Proposition 
\ref{prop:approx-log}.}
{\em 
There is an algorithm that,
given a positive integer $a$, encoded in binary, and given
a positive integer $j$, encoded in unary, computes a rational value $v_a$,
such that $$| \log(a) - v_a | < 2^{-j}$$
The algorithm runs in time polynomial
in $j$ and $\log_2(a)$  (in the Turing model).
}

\begin{proof}
Recall the standard power series 
for the natural logarithm of $x+1$, which holds for 
any $x$ in the range $-1 < x < 1$:

\begin{equation}
\label{eq:log-series}
\log (x + 1) =  x - x^2/2 + x^3/3 - x^4/4 + \ldots =  \sum^{\infty}_{i=1} (-1)^{i+1} \frac{x^i}{i}
\end{equation}

Consider any positive integer $a$.   We can assume $a > 1$, wlog,
because otherwise  $\log(1) = 0$.
By examining the binary encoding of $a$, we can easily determine a
positive integer $m > 0$,
such that $2^m \leq a < 2^{m+1}$.    Note that 

$$\log(a) = \log(a/2^{m+1}) + \log(2) * (m+1)$$

Thus, to compute $\log(a)$ ``to within sufficient accuracy'', 
it suffices to compute both $\log(a/2^{m+1})$  and $\log(2)$, 
``to within sufficient accuracy''.   
But note that since  
$a \geq 2^m$, we have that $1/2 \leq a/2^{m+1} < 1$.    
Letting $y := (a/2^{m+1} -1)$,   since $-1/2 \leq y < 0$,
we have $y+1 = a/2^{m+1}$ is in the convergent region of the
series (\ref{eq:log-series}), so that:

$$\log(a/2^{m+1}) = \sum^{\infty}_{i=1} (-1)^{i+1} \frac{y^i}{i}$$

Now, note that since  $-1/2 < y < 0$,  we have 
$|\sum^{\infty}_{i=k+1} (-1)^{i+1} \frac{y^i}{i}| < 
\sum^{\infty}_{i=k+1} (1/2)^i \leq (1/2^k)$.
Thus, to compute an approximant  $v'$ for $\log(a/2^{m+1})$ such that $| \log(a/2^{m+1}) - v' | < (1/2^k)$,
all we have to do is to compute the first $k+1$ terms of the series 
$\sum^{\infty}_{i=1} (-1)^{i+1} \frac{y^i}{i}$.    This we can easily do in time polynomial in $k$ and in the encoding
size of $a$, by just carrying out the arithmetic.
(The encoding size of the numbers calculated this way will not get more than polynomially large in the 
encoding size of $a$ and $k$: roughly each term has encoding size at most $k*size(a)$,  so the encoding
size of the sum is at most $\log_2(k) * k * size(a)$.)

Next, we need to compute an approximation of $\log(2)$.
To do this, we use a well-known alternative series derivable from
(\ref{eq:log-series}): we have  $\log(2) = - \log(1/2) = -\log(1 +  -(1/2)) = 
\sum^{\infty}_{i=1}  \frac{1}{i 2^i}$.
Again, we have  $\sum^{\infty}_{i=k+1}  \frac{1}{i 2^i} \leq 1/2^{k}$.
Thus, we can  compute an approximant  $v''$ of $\log(2)$, such that $| v'' - \log(2) | < 2^{-j}$,
by just computing the first $j+1$ terms of the series $\sum^{\infty}_{i=1}  \frac{1}{i 2^i}$.

We would like to combine a suitable approximation $v'$ of $\log(a/2^{m+1})$ and a suitable approximation 
$v''$ of $\log(2)$ to get an approximation of $\log(a) = \log(a/2^{m+1}) + \log(2)(m+1)$, to 
within a desired additive error $2^{-j}$ in time polynomial in the encoding size of $a$ and in $j$. 

We only need to observe that if $| \log(a/2^{m+1}) - v' | \leq 2^{-(j+1)}$   and $| \log(2) - v'' | < 2^{-(j+1)}/(m+1)$,
and if we let $v_a := (v'-v''(m+1))$,  then 

\begin{eqnarray*}
| \log(a) -  v_a |  & = & | (\log(a/2^{m+1}) + \log(2) (m+1)) - (v' + v''(m+1)) | \\
& = &  | (\log(a/2^{m+1})  - v')  +  (m+1) (\log(2) - v'') | \\ 
& \leq  &  | (\log(a/2^{m+1})  - v') | +  (m+1) |\log(2) - v''| \\
& \leq &   2^{-(j+1)} + 2^{-(j+1)} = 2^{-j}
\end{eqnarray*}

Thus, given positive integer $a$ in binary,  
we can compute a $2^{-j}$-approximation, $v_a$,  of $\log(a)$ in 
time polynomial in the encoding size of $a$ and in $j$, in the Turing model of computation.
 \end{proof}


\begin{thebibliography}{10}

\bibitem{ABKM06}
E.~Allender, P.~B\"{u}rgisser, J.~Kjeldgaard-Pedersen, and P.~B. Miltersen.
\newblock On the complexity of numerical analysis.
\newblock {\em SIAM J. Comput.}, 38(5):1987--2006, 2009.

\bibitem{BachDS93}
E. Bach, J. Driscoll, and J. Shallit.
\newblock  Factor refinement.
\newblock {\em Journal of Algorithms}, volume 15, pages 199-222,
1993.

\bibitem{Baker-abc-1998}
A.~Baker.
\newblock Logarithmic forms and the {$abc$}-conjecture.
\newblock In {\em Number theory ({E}ger, 1996)}, pages 37--44. de Gruyter,
  Berlin, 1998.

\bibitem{BakerWust93}
A.~Baker and G.~W\"{u}stholz.
\newblock Logarithmic forms and group varieties.
\newblock {\em J. reine angew. Math.}, 442:19--62, 1993.

\bibitem{baker-wustholz-survey}
A.~Baker and G.~W{\"u}stholz.
\newblock Number theory, transcendence and {D}iophantine geometry in the next
  millennium.
\newblock In {\em Mathematics: frontiers and perspectives}, pages 1--12. Amer.
  Math. Soc., Providence, RI, 2000.

\bibitem{Baker2007}
A. Baker.
\newblock On an arithmetical function associated with the {$abc$}-conjecture.
\newblock In {\em Diophantine geometry}, volume~4 of {\em CRM Series}, pages
  25--33. Ed. Norm., Pisa, 2007.

\bibitem{baker-wustholz-book07}
A. Baker and G. W\"{u}stholz.
\newblock {\em Logarithmic Forms and Diophantine Geometry}.
\newblock New Mathematical Monographs: 9. Cambridge University Press, 2007.

\bibitem{bernstein05}
D.~J. Bernstein.
\newblock Factoring into coprimes is essentially linear time.
\newblock {\em Journal of Algorithms}, volume 54, pages 1--30, 2005. 

\bibitem{DEKM99}
R.~Durbin, S.~R. Eddy, A.~Krogh, and G.~Mitchison.
\newblock {\em Biological Sequence Analysis: Probabilistic models of Proteins
  and Nucleic Acids}.
\newblock Cambridge U. Press, 1999.

\bibitem{ESY12}
K.~Etessami, A.~Stewart, and M.~Yannakakis.
\newblock Polynomial-time algorithms for multi-type branching processes and
  stochastic context-free grammars.
\newblock In {\em Proc. 44th ACM Symposium on Theory of Computing (STOC)},
  2012.
\newblock (Full version on ArXiv:1201.2374).

\bibitem{GGJ76}
M.~R. Garey, R.~L. Graham, and D.~S. Johnson.
\newblock Some {NP}-complete geometric problems.
\newblock In {\em Proc. 8th ACM Symp. on Theory of Computing (STOC)}, pages
  10--22, 1976.

\bibitem{granville-abc-2002}
A.~Granville and T.~J. Tucker.
\newblock It's as easy as abc.
\newblock {\em Notices of the AMS}, 49(10):1224--1231, 2002.

\bibitem{KI03}
V.~Kabanets and R.~Impagliazzo.
\newblock Derandomizing polynomial identity tests means proving circuit lower
  bounds.
\newblock In {\em Proc. 35th ACM Symp. on Theory of Computing (STOC)}, pages
  355--364, 2003.

\bibitem{Knuth77}
D.~Knuth.
\newblock A generalization of Dijkstra's algorithm.
\newblock {\em Information Processing Letters}, 6:1--5, 1977.

\bibitem{LangDio78}
S.~Lang.
\newblock {\em Eliptic Curves: Diophantine Analysis}.
\newblock Springer, 1978.

\bibitem{ManSch99}
C.~Manning and H.~Sch\"{u}tze.
\newblock {\em Foundations of Statistical Natural Language Processing}.
\newblock MIT Press, 1999.

\bibitem{Masser-abc}
D.~Masser.
\newblock Open problems.
\newblock In W.~W.~L. Chen, editor, {\em Proc. Symp. Analytic Number Theory},
  London Math. Soc. Lecture Notes, Ser. 96. Cambridge U. Press, 1985.

\bibitem{Matveev-2000I}
E.~M. Matveev.
\newblock An explicit lower bound for a homogeneous rational linear form in
  logarithms of algebraic numbers. {I}.
\newblock {\em Izv. Ross. Akad. Nauk Ser. Mat.}, 62(4):81--136, 1998.
\newblock (Russian original. English translation in {\em Izvestiya Math.} 62(2)
  (1998), 723--772.).

\bibitem{Matveev-2000II}
E.~M. Matveev.
\newblock An explicit lower bound for a homogeneous rational linear form in
  logarithms of algebraic numbers. {II}.
\newblock {\em Izv. Ross. Akad. Nauk Ser. Mat.}, 64(6):125--180, 2000.
\newblock (Russian original. English translation in {\em Izv. Math.} 64 (2000),
  no. 6, 1217--1269.).

\bibitem{NedSat08b}
M.-J. Nederhof and G.~Satta.
\newblock Probabilistic parsing.
\newblock {\em New Developments in Formal Languages and Applications},
  113:229--258, 2008.

\bibitem{Nesterenko-2003}
Y. Nesterenko.
\newblock Linear forms in logarithms of rational numbers.
\newblock In {\em Diophantine approximation ({C}etraro, 2000)}, volume 1819 of
  {\em Lecture Notes in Math.}, pages 53--106. Springer, Berlin, 2003.

\bibitem{Oesterle-abc}
J.~Oesterl\`{e}.
\newblock Nouvelles approches du "theoreme" de fermat.
\newblock {\em Asterisque}, 161-2:165--186, 1988.

\bibitem{Schinzel00}
A.~Schinzel.
\newblock {\em Polynomials with special regard to reducibility}.
\newblock Cambridge U. Press, 2000.


\bibitem{Shub-smalefest-93}
M.~Shub.
\newblock{Some remarks on Bezout's theorem and complexity theory.}
\newblock In {\em From Topology to Computation: Proceedings of the Smalefest},
M. Hirsch, J. Marsden, and M. Shub (Eds.), pp. 443-455, Springer-Verlag, 1993.

\bibitem{Waldschmidt-93}
M. Waldschmidt.
\newblock Minorations de combinaisons lin\'eaires de logarithmes de nombres
  alg\'ebriques.
\newblock {\em Canad. J. Math.}, 45(1):176--224, 1993.

\bibitem{Waldschmidt-survey04}
M. Waldschmidt.
\newblock Open {D}iophantine problems.
\newblock {\em Mosc. Math. J.}, 4(1):245--305, 312, 2004.

\end{thebibliography}
\end{document}